\documentclass[11pt,a4paper]{article}

\usepackage[T1]{fontenc}
\usepackage{lmodern}
\usepackage[margin=1in]{geometry}
\usepackage{amsmath,amssymb,amsthm,graphicx,placeins}
\graphicspath{{figures/}}
\usepackage{booktabs}
\usepackage{csquotes}
\usepackage{microtype}
\usepackage[authoryear,round]{natbib}
\usepackage{authblk}
\usepackage[hidelinks]{hyperref}
\hypersetup{
  pdftitle={The Complexity of Value Optimization in Sporadic Task Systems},
  pdfauthor={Christopher Harbourne; Nathan Fisher; Kecheng Yang; Pontus Ekberg; Sanjoy Baruah; Zhishan Guo}
}

\theoremstyle{plain}
\newtheorem{theorem}{Theorem}
\newtheorem{lemma}{Lemma}

\theoremstyle{definition}
\newtheorem{definition}{Definition}

\newcommand{\dbf}{\operatorname{dbf}}

\newcommand{\VTS}{\text{\normalfont\textsc{vts}}}

\title{The Complexity of Value Optimization in Sporadic Task Systems}
\author[1]{Christopher Harbourne}
\author[1]{Nathan Fisher}
\author[2]{Kecheng Yang}
\author[3]{Pontus Ekberg}
\author[4]{Sanjoy Baruah}
\author[5]{Zhishan Guo}
\affil[1]{Wayne State University, Detroit, Michigan, USA\protect\\
  \href{mailto:gf4350@wayne.edu}{\nolinkurl{gf4350@wayne.edu}}, \href{mailto:fishern@wayne.edu}{\nolinkurl{fishern@wayne.edu}}}
\affil[2]{Texas State University, San Marcos, Texas, USA\protect\\
  \href{mailto:yangk@txstate.edu}{\nolinkurl{yangk@txstate.edu}}}
\affil[3]{Uppsala University, Uppsala, Sweden\protect\\
  \href{mailto:pontus.ekberg@it.uu.se}{\nolinkurl{pontus.ekberg@it.uu.se}}}
\affil[4]{Washington University in St. Louis, St. Louis, Missouri, USA\protect\\
  \href{mailto:baruah@wustl.edu}{\nolinkurl{baruah@wustl.edu}}}
\affil[5]{North Carolina State University, Raleigh, North Carolina, USA\protect\\
  \href{mailto:zguo32@ncsu.edu}{\nolinkurl{zguo32@ncsu.edu}}}
\date{}

\begin{document}
\maketitle

\begin{abstract}
Value-based Task Selection (\(\VTS\)) asks whether a real-time system can select tasks whose total value reaches a required target while preserving every deadline. We study its threshold decision form for finite collections of independent sporadic tasks scheduled preemptively on one processor. We prove that this problem is \(\Sigma_2^P\)-complete for both constrained and arbitrary deadlines, even when all numerical parameters are encoded in unary. The existential choice is the selected task set, while the universal condition requires every processor-demand window to remain feasible. Our hardness reduction already uses constrained-deadline tasks. It represents universal assignments by Chinese Remainder Theorem timestamps and uses polynomial-weight local selector locks to enforce a valid existential assignment. The reduction produces polynomially many tasks, and every generated execution time, deadline, period, task value, and target is polynomially bounded in the source-instance size.
\end{abstract}

\noindent\textbf{Keywords:} real-time scheduling, sporadic tasks, task selection, value optimization, demand bound function, polynomial hierarchy

\section{Introduction}

Choosing a valuable task set and deciding whether a fixed task set is feasible are often treated as separate questions. We study the combined problem: given a finite collection of independent sporadic tasks on one preemptive processor, each with an associated value, can one select a subset whose total value reaches a prescribed target and whose jobs all meet their deadlines under every legal release sequence? The selection is made offline, with the entire collection available. Once a task is selected, every job of that task must retain its hard-deadline guarantee. We show that this task-selection problem, which we call the \emph{Value-based Task Selection (VTS)} problem, has the existential-universal structure of the second level of the polynomial hierarchy.

The central result is that VTS is \(\Sigma_2^P\)-complete for both constrained and arbitrary deadlines, even when all numerical parameters are encoded in unary. Optional selector tasks encode an existential choice, while demand windows expose possible counterexamples to that choice. For a well-formed selection, the construction creates more required execution than available processor time in a window exactly when the encoded choice has a logical counterexample. The remainder of this section develops this interpretation before stating the contributions and organization of the paper.

\subsection{Motivation}

Real-time systems may contain more useful functionality than their available processing capacity can support. A concrete design-time setting is the configuration of an embedded platform whose recurrent tasks implement optional monitoring, diagnostic, or quality-enhancement functions. Enabling a feature selects its entire sporadic task: every job must meet its hard deadline under every legal release sequence. Disabling the feature removes the complete recurrent workload. Selecting more functionality increases total value, but the resulting workload must still meet every hard deadline.

Utility-driven task selection also appears in cloud computing. \citet{li2016utility} select periodic tasks and their parallelization options to maximize total utility on multiple processing units, under implicit deadlines. That model differs from our sequential sporadic tasks on one processor, but it provides a concrete example of choosing which recurrent workloads to retain under limited processing capacity.

For implicit-deadline tasks on one preemptive processor, each selected task consumes a single utilization share, giving the scalar-capacity structure of 0-1 knapsack. With constrained deadlines, the same selected set must instead satisfy a family of processor-demand inequalities, one for each relevant interval length. The standard single-capacity knapsack formulation therefore does not directly supply algorithms or approximation guarantees for this setting.

This task-selection problem also connects to multiple-implementation selection (MIS), where exactly one implementation is chosen per functionality. If MIS permits a disabled option with zero execution and zero value, each task becomes a two-option choice: select it or disable its functionality. Section~\ref{sec:related-work} gives this embedding and the resulting completeness classification under the same scheduling assumptions.

\subsection{Problem Overview}

Following the sporadic task model introduced by \citet{mok1983fundamental}, the input is an indexed multiset of tasks \[ \tau_i=(C_i,D_i,T_i,V_i), \] where \(C_i\), \(D_i\), and \(T_i\) are execution time, relative deadline, and minimum inter-arrival separation, and \(V_i\) is the value of selecting \(\tau_i\). We consider independent tasks on one preemptive unit-speed processor. Earliest deadline first (EDF) is optimal under these assumptions. The formal model appears in Section~\ref{sec:model}.

We study the threshold decision form because completeness is defined for decision problems: the VTS problem determines whether the maximum feasible value is at least \(K\). For a selected indexed submultiset \(A\), write \[ V(A)=\sum_{\tau_i\in A}V_i. \] In the implicit-deadline case \(D_i=T_i\), EDF feasibility is equivalent to \[ \sum_{\tau_i\in A}\frac{C_i}{T_i}\le1. \] Therefore tasks act as knapsack items, their values are \(V_i\), and their sizes are utilization shares \(C_i/T_i\). This is the classical 0-1 knapsack structure, for which exact and approximation algorithms are well studied \citep{dantzig1957discrete,garey1979computers,kellerer2004knapsack}.

Constrained deadlines replace this single capacity constraint with interval-dependent demand constraints. The \emph{demand bound function} (DBF), denoted \(\dbf(A,\ell)\), is the maximum execution demand of jobs of \(A\) that can both arrive and have their deadlines within an interval of length \(\ell\). Feasibility requires \(\dbf(A,\ell)\le\ell\) for every \(\ell>0\) (e.g., see~\cite{baruah1990sporadic}); Section~\ref{sec:model} gives the formula and exact criterion.

\subsection[Why the Problem Has a Sigma-2-P Shape]{Why the Problem Has a
\(\Sigma_2^P\) Shape}

Throughout this explanation, a proposed solution or counterexample is called \emph{short} if its encoding length is bounded by a polynomial in the input size. The predicate \(R\) is computable in polynomial time. The polynomial hierarchy extends NP and coNP through successive levels with a fixed number of alternations between existential and universal choices \citep{stockmeyer1976ph, wrathall1976complete}. For an NP problem, an instance $I$'s membership is confirmed by exhibiting one short proposed solution \(w\) and checking it in polynomial time:
\(\exists w:\ R(I,w)=1\).
For a coNP problem, membership is confirmed only if every short proposed counterexample \(w\) fails to break it:
\(\forall w:\ R(I,w)=1\).
The second existential level, \(\Sigma_2^P\), combines these two ideas:
\[ \exists w_1\;\forall w_2:\quad R(I, w_1, w_2)=1.\]
In plain terms, there is a short candidate solution that remains valid under every short challenge. For VTS, after choosing tasks whose value reaches the target, one question remains: can this set meet every deadline? Imagine being able to hand the set to a checker and get a correct yes-or-no answer. Complexity theory calls this assumed checker an \emph{oracle}. We could ask \enquote{Is the set feasible?} or instead ask \enquote{Is there a demand window that overloads the processor?} A yes to one means a no to the other. The feasibility question is in coNP; the overload question is in NP. Either answer therefore lets us judge our chosen set, giving the two equivalent descriptions
\[ \Sigma_2^P=\mathrm{NP}^{\mathrm{coNP}}=\mathrm{NP}^{\mathrm{NP}}. \]
\citet{ekberg2021partitioned} discuss these views and their relevance to real-time scheduling.

Figure~\ref{fig:ph-orientation} illustrates the levels of the polynomial hierarchy and how each one builds on the levels below it.

\begin{figure}[ht]
\centering
\includegraphics[width=0.97\linewidth]{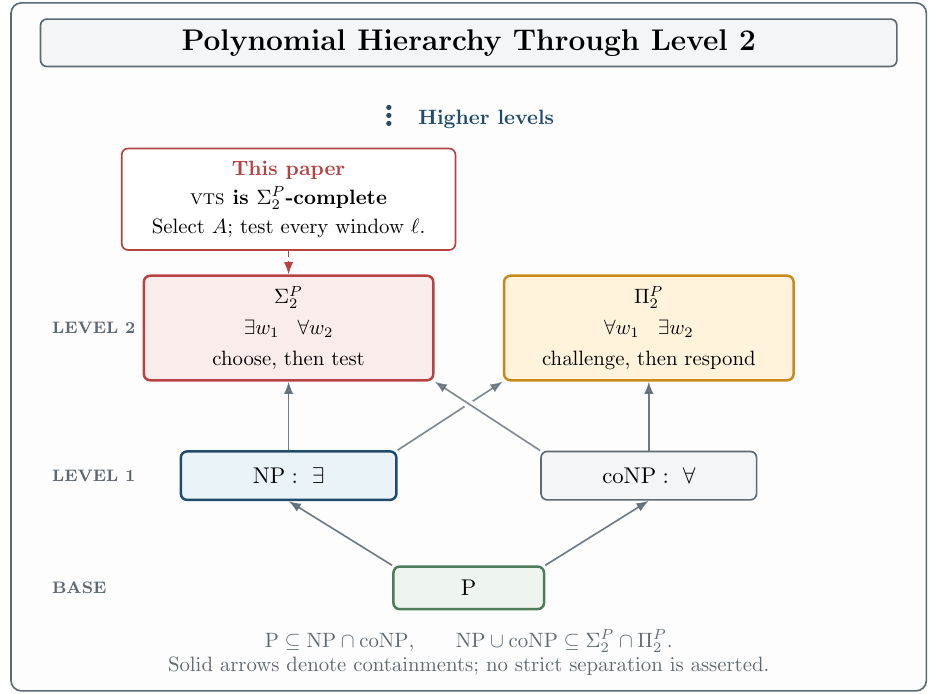}
\caption{The polynomial hierarchy through its second level. The enclosing
boundary places all displayed classes inside PH, and the continuation cue
indicates higher levels. Solid arrows show class containments; no strict
separation is asserted. The separate callout identifies this paper's result:
\(\VTS\) is \(\Sigma_2^P\)-complete. Its dashed pointer identifies the
relevant class and does not represent a class-containment relation.}
\label{fig:ph-orientation}
\end{figure}

The order of the choices distinguishes the two branches: \(\Sigma_2^P\) chooses first and then tests every response; \(\Pi_2^P\) reverses that order. VTS follows the first pattern.

The \(\VTS\) decision problem has the form:
\[ \exists A\subseteq\Gamma:\quad \bigl[V(A)\ge K\bigr]\ \land\ \bigl[\forall\ell>0:\ \dbf(A,\ell)\le\ell\bigr]. \]
The selected set \(A\) is the existential choice. Once \(A\) is fixed, one interval satisfying \(\dbf(A,\ell)>\ell\) is a counterexample to its feasibility. The display describes the scheduling condition; Section~\ref{sec:membership} establishes the polynomial-length counterexample bound needed for the complexity classification.

A \(\Sigma_2^P\)-completeness result is therefore stronger than separate NP- and coNP-hardness results. Unless the polynomial hierarchy collapses, the problem lies in neither NP nor coNP.

\subsection{Contributions}
\begin{itemize}
    \item We formalize the threshold decision form of Value-based Task Selection for sporadic tasks on one preemptive processor: every selected task executes in full, and the selected set must reach a prescribed value target.

    \item We prove that \(\VTS\) is \(\Sigma_2^P\)-complete for both constrained and arbitrary deadlines. The constrained-deadline hardness construction uses demand windows indexed via the Chinese Remainder Theorem (CRT) to represent universal responses and a polynomial-weight local selector lock to enforce one existential truth choice per variable.

    \item We prove that the construction produces polynomially many tasks and polynomially bounded numerical parameters. As a result, completeness also holds when all numerical parameters are encoded in unary.
\end{itemize}

\subsection{Organization}
Section~\ref{sec:related-work} reviews related work, and Section~\ref{sec:model} defines the model and decision problem. Section~\ref{sec:membership} proves membership in \(\Sigma_2^P\). Section~\ref{sec:reduction-overview} introduces the quantified source problem, summarizes the complete construction, and states the reduction invariant. Sections~\ref{sec:gadgets} and~\ref{sec:selector-lock} develop the demand gadgets and selector lock. Sections~\ref{sec:construction} and~\ref{sec:correctness} give the detailed construction and correctness proof, and Section~\ref{sec:boundedness} proves the task-count and numerical-parameter bounds. Section~\ref{sec:main-result} states the completeness theorem. Sections~\ref{sec:discussion} and~\ref{sec:conclusion} discuss its implications and conclude.

\section{Related Work}
\label{sec:related-work}

Four strands of prior work provide the closest context for this problem: real-time admission control, optional or imprecise execution, exact uniprocessor feasibility analysis, and polynomial-hierarchy classifications in real-time scheduling. They differ in what is selected, when the selection is made, or which scheduling platform is studied.

\subsection{Admission Control and Value-Aware Scheduling}

Classical real-time admission control decides whether newly arriving work can be accepted without invalidating guarantees already made to admitted work \citep{chetto1989scheduling,andersson2005admission}. The problem studied here instead receives the entire task collection at design time and selects a feasible subset whose fixed task values reach a target.

Value-based overload scheduling treats completed work as an optimization objective. \citet{burns2000meaning} survey the role of value in flexible real-time scheduling, while \citet{koren1992dover} study online utility-maximizing admission under overload. More recent work studies value-aware scheduling for transportation applications and response-time-dependent utility accrual \citep{yang2026value,chen2025utility}. These models provide related utility objectives; they do not establish the complexity of the offline whole-task selection problem studied here.

Knapsack formulations also arise in other real-time resource decisions. \citet{mejia2004adaptive} choose among discrete processor speeds to improve energy savings under timing constraints. \citet{goossens2013partitioned} use a knapsack subproblem to bound a mode-transition delay in partitioned multimode scheduling. The objectives and scheduling models differ from ours; these examples illustrate the use of scalar-capacity selection subproblems in real-time analysis.

Multiple-choice knapsack selects exactly one alternative from each group \citep{sinha1979multiple}. A related multiple-implementation selection (MIS) model chooses exactly one implementation per functionality. For example, alternative classifiers for the same classification task may trade execution demand against utility. For MIS models permitting an explicit disabled implementation, task selection is the special case with two alternatives per functionality: the given task, or a disabled option with zero execution and zero value. The latter is an option of the extended model, not a task in our positive-parameter domain. For explicitly listed alternatives under the same scheduling assumptions, this MIS threshold decision problem is \(\Sigma_2^P\)-complete for both constrained and arbitrary deadlines. Hardness follows from the two-option embedding. For membership, guess exactly one implementation per functionality, check the total value, and apply the feasibility check to the selected implementations (see Theorem~\ref{thm:membership}).

\subsection{Optional and Imprecise Execution}

Imprecise-computation and reward-based models divide work into mandatory and optional portions, allowing result quality to improve as optional execution is completed \citep{liu1991imprecise,shih1996online,aydin2001reward}. These models optimize how much optional service to execute after preserving mandatory timing guarantees. In contrast, our decision is made at task granularity: a selected sporadic task executes in full, and all of its jobs must meet hard deadlines.

\subsection{Feasibility and Complexity of Uniprocessor Scheduling}

For preemptive uniprocessor scheduling of constrained-deadline sporadic tasks, processor demand gives a necessary and sufficient feasibility condition \citep{baruah1990sporadic}. The complexity of evaluating such conditions was clarified in stages. Eisenbrand and Rothvo{\ss} proved weak coNP-hardness for EDF schedulability of synchronous periodic task systems, and Ekberg and Yi later established strong coNP-completeness for constrained-deadline sporadic feasibility \citep{eisenbrand2010edf,ekberg2015constrained}. Related difficulty also appears under static-priority scheduling: exact response-time computation is NP-hard \citep{eisenbrand2008response}.

Our reduction builds on the modular structure used in the constrained-deadline feasibility result. CRT-based constructions identify demand windows whose residue patterns represent Boolean assignments; the constructed deadline offsets turn matching patterns into controlled processor-demand spikes.

\subsection{Polynomial-Hierarchy Results in Real-Time Scheduling}

\citet{ekberg2021partitioned} classify a broad family of partitioned multiprocessor schedulability problems rather than treating them uniformly as
NP-hard. For asynchronous periodic task systems scheduled by EDF, they prove that the constrained- and arbitrary-deadline partitioning problems are
\(\Sigma_2^P\)-complete, for both a constant and an arbitrary number of processors. For fixed-priority scheduling, several corresponding problems are
\(\Sigma_2^P\)-hard but have only a \(\Sigma_3^P\) upper bound; exact classification therefore remains open in those cases. They also identify
fixed-priority implicit-deadline cases whose classification depends on whether the priority order is specified.

\citet{agrawal2023rethinking} relate exact, pseudo-polynomial, and approximation-based notions of tractability in real-time schedulability analysis. They also analyze a memory-constrained task-selection problem to illustrate limits of approximation schemes. That selection model and objective differ from the hard-deadline value-threshold problem studied here.

Our result exhibits the same higher-level quantifier structure without a partitioning decision or multiple processors. Here the existential choice is
which tasks to select, and the universal condition requires the resulting single-processor task set to satisfy every processor-demand constraint.

\section{Model and Problem Definition}
\label{sec:model}

This section fixes the model before any complexity argument begins. It defines sporadic tasks, deadline restrictions, and the assumptions under which EDF is optimal, introduces the demand bound function and its exact feasibility criterion, and then defines the value-threshold decision problem.

\subsection{Task Model}
\begin{definition}[Sporadic Task]
    A task is a tuple \[ \tau_i=(C_i, D_i, T_i, V_i) \] where \(C_i\) is execution time, \(D_i\) is relative deadline, \(T_i\) is minimum inter-arrival separation, and \(V_i\) is the value obtained if \(\tau_i\) is selected. Each job requires at most \(C_i\) units of execution and has an absolute deadline \(D_i\) time units after its release. Consecutive releases are separated by at least \(T_i\) time units. All four parameters are positive integers, with \(C_i\le D_i\).
\end{definition}

Deadlines are \emph{implicit} if \(D_i=T_i\) for every task and \emph{constrained} if \(D_i\le T_i\) for every task. The \emph{arbitrary-deadline} model imposes no relation between \(D_i\) and \(T_i\). The hardness construction below produces only constrained-deadline tasks; however, since the arbitrary-deadline model is a generalization of constrained-deadline model, the hardness result for constrained-deadlines immediately implies a hardness result for arbitrary deadlines.

A release sequence \(r_{i,0},r_{i,1},\ldots\) is \emph{legal} for \(\tau_i\) if \(r_{i,0}\ge0\) and \(r_{i,k+1}-r_{i,k}\ge T_i\) for every consecutive pair of releases. A collection of release sequences is legal if this condition holds for every task.

Tasks are independent, jobs are fully preemptible, and the processor has unit speed. A task system is \emph{feasible} if some scheduling algorithm meets every deadline for every legal collection of release sequences, with each job requiring at most its task's execution-time bound. Earliest deadline first is optimal under these assumptions \citep{dertouzos1974control}, so feasibility and EDF-schedulability are equivalent.

\subsection{Demand Bound Function}
For a selected indexed submultiset \(A\), define its \emph{demand bound function} (DBF) by
\[ \dbf(A,\ell) = \sum_{\tau_i\in A} \max\left\{0,\left\lfloor\frac{\ell-D_i}{T_i}\right\rfloor+1 \right\}C_i. \]
Throughout the paper, \(A\) denotes a selected task set. For each task, the floor term counts the largest number of jobs that can both arrive and have their deadlines within a window of length \(\ell\). Multiplying that count by \(C_i\) gives the corresponding execution demand. Summing over the selected tasks therefore measures the maximum processor demand that must fit inside the window.

The utilization of \(A\) is
\[ U(A)=\sum_{\tau_i\in A}\frac{C_i}{T_i}. \]
Every term of \(\dbf(A, \ell)\) and \(U(A)\) is nonnegative, so demand and utilization are monotone in the selected set. In particular, removing tasks from a feasible set preserves feasibility. Therefore additional tasks can increase total value while making the selected set infeasible.

\subsection{Feasibility Criterion}
\begin{theorem}[Demand-bound feasibility criterion;
\citet{baruah1990sporadic}]
\label{thm:dbf-criterion}
A selected sporadic task system \(A\) is feasible on a preemptive uniprocessor if and only if
\[ U(A) \le 1 \quad\text{and}\quad \dbf(A, \ell)\le \ell \text{ for all } \ell>0. \]
\end{theorem}

The DBF is a right-continuous step function of the interval length \(\ell\). Its possible discontinuities occur at \(D_i+kT_i\), for \(\tau_i\in A\) and integers \(k\ge0\). These are deadlines in a synchronous minimum-separation release pattern; an arbitrary sporadic release sequence can have different absolute deadlines.

\begin{lemma}[Sufficiency of DBF discontinuities; cf.\ \citet{baruah1990sporadic}]
\label{lem:jump}
If \(\dbf(A,\ell)>\ell\) for some \(\ell>0\), then there exist \(\tau_i\in A\), an integer \(k\ge0\), and \(\ell'=D_i+kT_i\le\ell\) such that \(\dbf(A,\ell')>\ell'\).
\end{lemma}

\subsection[VTS Decision Problem]{\(\VTS\) Decision Problem}

\begin{definition}[\(\VTS\)]
An instance consists of a finite indexed multiset \(\Gamma\) of sporadic tasks and a nonnegative integer value target \(K\). The question is whether there exists a selected indexed submultiset \(A\subseteq\Gamma\) such that
\[
    V(A)=\sum_{\tau_i\in A}V_i\ge K
    \qquad\text{and}\qquad
    \dbf(A,\ell)\le\ell\quad\text{for every }\ell>0.
\]
The constrained-deadline restriction additionally requires \(D_i\le T_i\) for every input task.
\end{definition}

The DBF condition already implies \(U(A)\le1\), since
\[
    \lim_{\ell\to\infty}\frac{\dbf(A,\ell)}{\ell}=U(A).
\]

A selection \(A\) \emph{meets the value target} when \(V(A)\ge K\). This value condition alone does not imply that \(A\) is feasible.

With the decision problem fixed, the next section proves the upper bound needed
for the completeness result.

\section[Membership in Sigma-2-P]{Membership in \(\Sigma_2^P\)}
\label{sec:membership}

The upper-bound argument separates the same two roles used in the hardness construction. A selection vector supplies the existential choice. Once that choice is fixed, one polynomial-length demand window can verify infeasibility, so feasibility is a coNP condition.

\begin{theorem}
\label{thm:membership}
\(\VTS\in \Sigma_2^P\) for constrained and arbitrary deadline sporadic task systems.
\end{theorem}

\begin{proof}
    Consider a candidate selection vector that specifies a selected indexed submultiset \(A\). Its value \(V(A)\) is computable in polynomial time. The remaining question is whether this fixed task set is feasible. The following argument verifies infeasibility in NP for any sporadic task system, using Theorem~\ref{thm:dbf-criterion}. If \(U(A)>1\), infeasibility is direct. The empty selection is trivially feasible. For nonempty \(A\) with \(U(A)\le 1\), let 
    \[ P(A)=\operatorname{lcm}\{T_i:\tau_i\in A\}, \qquad D_{\max}(A)=\max_{\tau_i\in A}D_i.\]
    For every \(\ell>P(A)+D_{\max}(A)\), increasing the window length by one hyper-period \(P(A)\) adds exactly \(P(A)C_i/T_i\) demand from each task \(\tau_i\in A\). Therefore,
    \[ \dbf(A,\ell)=\dbf(A,\ell-P(A))+P(A)U(A).\]
    Because \(U(A)\le 1\), a violation at \(\ell\) implies a violation at
    \(\ell-P(A)\). Repeating this shift shows that, if \(A\) is infeasible,
    some violating window satisfies \(\ell\le P(A)+D_{\max}(A)\). By
    Lemma~\ref{lem:jump}, it may be chosen at a DBF discontinuity.
    Although \(P(A)\) can be exponentially large, its binary representation
    has at most
    \(1+\sum_{\tau_i\in A}\lceil\log_2 T_i\rceil\) bits. Since
    \(D_{\max}(A)\) is an input parameter, every
    \(\ell\le P(A)+D_{\max}(A)\) also has polynomial binary length. Given this
    binary-encoded \(\ell\), checking \(\dbf(A,\ell)>\ell\) takes polynomial
    time. Therefore infeasibility has an NP verification, and feasibility is
    in coNP. Once a selection \(A\) is chosen, the combined condition that
    \(V(A)\ge K\) and \(A\) is feasible is therefore coNP-decidable. The
    \(\VTS\) decision problem asks whether such a selection exists,
    placing it in \(\mathrm{NP}^{\mathrm{coNP}}=\Sigma_2^P\).
\end{proof}

\section{Reduction Overview}
\label{sec:reduction-overview}

We now introduce the quantified source problem and summarize the construction and its proof obligations.

\subsection{Source Problem and Reduction Idea}

We reduce from the standard \(\Sigma_2^P\)-complete \(\exists\forall\)-3DNF problem \citep{stockmeyer1976ph,wrathall1976complete,schaefer2002compendium}. Its instances have the form \[ \exists X\forall Y\,\varphi(X,Y).\] Here \(X=\{x_1,\ldots,x_n\}\), \(Y=\{y_1,\ldots,y_u\}\), and \(\varphi\) is a three-literal disjunctive-normal-form (3-DNF) formula with exactly three literals over distinct variables in each term. The question is whether some assignment \(\alpha\) to \(X\) makes \(\varphi(\alpha,\beta)\) true for every assignment \(\beta\) to \(Y\). Under the quantifier interpretation, the existential player commits to \(\alpha\) first. That choice must survive every response \(\beta\) of the universal player. Appendix~\ref{app:normal-form} justifies this normal form.

Fix an existential assignment \(\alpha\). To show that \(\alpha\) is not winning, it is enough to find one universal assignment \(\beta\) for which \(\varphi(\alpha, \beta)\) is false. We therefore work with its negation, a three-literal conjunctive-normal-form (3-CNF) formula:
\[ F(X,Y)=\neg\varphi(X,Y)=\mathcal C_1\wedge\cdots\wedge\mathcal C_m. \]
We write \(\mathcal C_k\) for a clause, reserving \(C_i\) for task execution time. Therefore \(F(\alpha,\beta)\) is true exactly when \(\beta\) is a counterexample to \(\alpha\), and \(\alpha\) is winning exactly when no \(\beta\) satisfies \(F(\alpha,\beta)\).

The reduction uses three representations of the counterexample condition: Boolean patterns, modular tests, and task demands. First, we identify the variable-consistency and clause-satisfaction patterns to recognize. Second, we convert these clause-satisfaction patterns into a number-theoretic representation as sets of congruences, which describe recurring sets of integers.  In this number-theoretic construction, truth assignments to boolean variables are represented by integer ``coordinates'' and the set of congruences enforce a consistent assignment of truth values in these clause-satisfaction patterns. Third, using the number-theoretic representation of the Boolean formula \(F\), we construct sporadic tasks whose combined demand encodes the congruences derived in the previous step. For each existential assignment \(\alpha\), the construction provides a corresponding task selection for which a deadline miss is possible exactly when some assignment \(\beta\) satisfies \(F(\alpha,\beta)\). The following paragraphs explain these translations.

We construct the entire task collection from \(F\) before any optional selection is made. The modular tests are imposed on a nonnegative integer demand-window index \(x\). Each coordinate has an assigned prime, and its value at \(x\) is the residue of \(x\) modulo that prime. Different coordinates use different primes. Boolean tests use residues \(0\) and \(p-1\) modulo a coordinate's prime \(p\). A \emph{pattern} specifies the coordinate residues required for one local test. The Chinese Remainder Theorem (CRT) combines these requirements into one congruence \(x\equiv a\pmod b\), where \(b\) is the product of the pattern's primes.

For example, the coordinate requirements \(x\equiv2\pmod3\) and \(x\equiv0\pmod5\) combine into \(x\equiv5\pmod{15}\). This illustrates one pattern; the construction applies the same translation to its clause and consistency tests.

A \emph{full gadget} \(G(a,b,w)\) converts the congruence and a weight \(w\) into \(b\) ordinary sporadic tasks. Their execution times, deadlines, and periods are chosen so that, at designated test windows, their combined demand equals a baseline contribution plus \(w\) when the index matches the congruence, and equals the baseline otherwise. Thus the pattern specifies a test, CRT supplies its arithmetic form, and the gadget supplies the tasks that implement it. All gadgets inspect the same index \(x\), so matching patterns agree on every shared coordinate.

The optional selectors are individual tasks. Choosing one of two selectors per existential variable fixes \(\alpha\). Each clause of \(F\) has a pattern for each satisfying assignment to its three variables, while
additional patterns provide validity and consistency tests. CRT also ensures that one index can realize one prescribed residue for each coordinate simultaneously. Arbitrary indices need not represent Boolean assignments; the validity, consistency, and clause tests together prevent these indices from creating false counterexamples.

The value target forces every task of every full gadget into the selected set. Utilization and demand constraints then enforce exactly one selector per existential variable. A selection with this structure is called \emph{well formed}. For such a selection encoding \(\alpha\), an overloaded window exists exactly when some universal assignment \(\beta\) satisfies \(F(\alpha,\beta)\). This describes how the negated matrix is used inside the reduction; the source problem is still \(\exists X\forall Y\,\varphi(X,Y)\), with its 3-DNF matrix. We assume \(n\ge1\), adding an unused existential variable if necessary, which leaves the answer unchanged.

Figure~\ref{fig:quantifier-mechanisms} gives the complete reduction at the level of roles rather than formulas. Later sections define the boxes in the
middle of the figure; the correctness proof shows that their composition has the behavior stated on the right.

\begin{figure}[ht]
\centering
\includegraphics[width=0.97\linewidth]{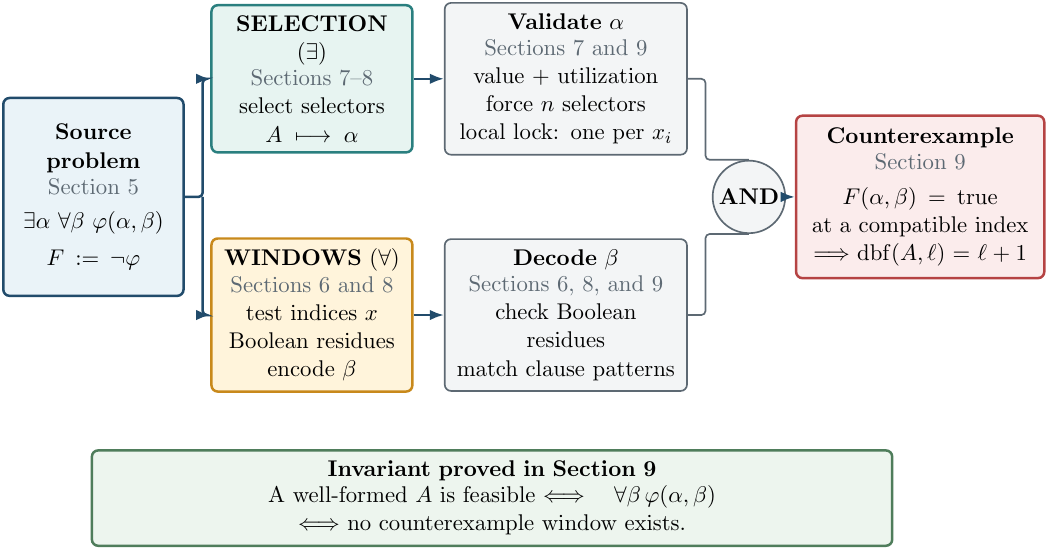}
\caption{Realization of the source quantifiers in the reduction. Optional task selections encode the existential assignment \(\alpha\). CRT ensures that every universal assignment \(\beta\) has a compatible window index. For a well-formed selection, any \(\beta\) satisfying \(F(\alpha,\beta)\) yields an overloaded window; the correctness proof also shows that an overloaded window implies such a counterexample.}
\label{fig:quantifier-mechanisms}
\end{figure}

\subsection{Construction at a Glance}
\label{sec:construction-preview}

Before developing the gadgets, we summarize the complete output instance. The detailed task parameters appear in Sections~\ref{sec:gadgets}--\ref{sec:construction}. A full gadget \(G(a,b,w)\) expands into \(b\) tasks and recognizes one residue pattern. The selector \(H(p,w)\) is a single optional task with a weight-\(w\) demand spike at the designated test windows when \(x\equiv p-1\pmod p\).

Each existential variable \(x_i\) has two coordinates, \(t_i\) and \(f_i\), with distinct primes. Each universal variable \(y_j\) has one coordinate with its own prime. An index is \emph{exact} for \(x_i\) when both of its coordinate residues are Boolean. For these exact tests, write \(1\) for a coordinate's maximum residue \(p-1\) and \(0\) for residue zero. In the order \((t_i,f_i)\), the pairs \(10\) and \(01\) encode true and false; \(11\) and \(00\) are the both-match and neither-match consistency tests. These pairs describe the window index's residues, not which tasks were selected.

The inventory uses three demand parameters: \(W_0\) is the anchor weight, \(q\) is the selector weight, and \(S\) is the local spike budget for one existential variable. These parameters control execution demand;
task values, which determine whether the selection reaches its target, are assigned separately below.

The output has the following five task families:
\begin{enumerate}
    \item To supply a common demand offset, one always-matching anchor \(G(0,1,W_0)\).

    \item To test consistency of the existential choices, four full helpers per \(x_i\), for exact states \(10,01,11,00\), with weights \(1,1,1,S\), respectively.

    \item To recognize Boolean universal values, two unit-weight full gadgets per \(y_j\), for residues \(0\) and \(p_{y_j}-1\).

    \item To recognize clause satisfaction, seven unit-weight full gadgets per clause \(\mathcal C_k\), one for each satisfying assignment to its three distinct variables.

    \item To choose the existential assignment, two optional selectors per \(x_i\): \(\tau_i^t=H(p_{t_i},q)\) and \(\tau_i^f=H(p_{f_i},q)\).
\end{enumerate}
Let \(\mathcal P\) be the weighted multiset of patterns in families 2--4. Each \(\rho\in\mathcal P\)  specifies a congruence \(a_\rho\pmod{b_\rho}\) and a weight \(w_\rho\). The detailed residue requirements appear in Section~\ref{sec:construction}. Let \(p_{\max}\) be the largest coordinate prime. The constants have the following roles:
\begin{itemize}
    \item The selector weight \(q\) and local spike budget \(S\) are \[ q=m+1,\qquad S=q+1. \] At most one satisfying pattern matches per clause, so the entire clause layer contributes at most \(m\) units of extra demand above its baseline. Choosing \(q>m\) makes a duplicate selector's penalty larger than that possible contribution. A matching selector together with a matching unit helper fills the budget \(S\).

    \item The total non-anchor mandatory weight and the anchor weight are \[ B_{\mathrm{mand}}=n(S+3)+2u+7m,\qquad W_0=1+qp_{\max}+S. \] Here \(B_{\mathrm{mand}}\) sums gadget weights, not individual tasks. The anchor supplies a fixed demand offset and helps ensure that the resulting tasks satisfy the constrained-deadline requirements.

    \item The common timing scale is \[ \Lambda=W_0+B_{\mathrm{mand}}+nq. \] This choice makes the utilization of all mandatory tasks and exactly \(n\) selected selectors equal to one.

    \item The forbidden extra-demand threshold is \[ h=W_0+nS+u+m. \] The terms correspond to the anchor, the \(n\) local budgets, the \(u\) universal-variable tests, and the \(m\) clause tests. Section~\ref{sec:gadgets} explains how the test-window lengths make this threshold exceed the capacity remaining beyond the baseline budget.
\end{itemize}
Using indexed multiset union, the complete task collection is
\[ \begin{aligned} \Gamma_{\mathrm{mand}} &=G(0,1,W_0)\uplus\Bigl(\biguplus_{\rho\in\mathcal P} G(a_\rho,b_\rho,w_\rho)\Bigr),\\
\Gamma_{\mathrm{opt}} &=\biguplus_{i=1}^{n}\{\tau_i^t,\tau_i^f\},\\ \Gamma&=\Gamma_{\mathrm{mand}}\uplus\Gamma_{\mathrm{opt}}. \end{aligned} \]
Let \(M=|\Gamma_{\mathrm{mand}}|\). Each task in \(\Gamma_{\mathrm{mand}}\) receives value \(V_0=2n+1\), each optional selector receives value \(1\), and the target is \(K=MV_0+n\). These values force every task in \(\Gamma_{\mathrm{mand}}\) into any selection meeting the target; the model itself has no mandatory-task constraint.

We call a selection \emph{well formed} if it contains every task in \(\Gamma_{\mathrm{mand}}\) and exactly one of \(\tau_i^t,\tau_i^f\) for each \(x_i\). Such a selection has value \(K\) and utilization one. The proof must establish that every feasible selection meeting the value target is well formed and that a well-formed selection is feasible exactly when its encoded existential assignment is winning.

\subsection{Reduction Roadmap}

The demand gadgets of Section~\ref{sec:gadgets} and the local selector tests of Section~\ref{sec:selector-lock} support four proof obligations:

\begin{enumerate}
    \item \emph{Force the required tasks and selector count.}
    The values chosen in Section~\ref{sec:construction} force every mandatory task and at least \(n\) selectors into a selection meeting the target. The utilization bound permits at most \(n\). Section~\ref{sec:correctness} therefore obtains exactly \(n\) selected selectors in total.

    \item \emph{Distribute the selectors one per variable.}
    The local tests of Section~\ref{sec:selector-lock} expose a duplicated choice without allowing a missing choice elsewhere to offset its extra demand. Section~\ref{sec:correctness} uses CRT to realize these tests at one common index and obtain an overloaded window. Therefore every feasible selection meeting the target is well formed and defines one assignment \(\alpha\).

    \item \emph{Relate feasibility to the source formula.}
    For a well-formed selection, Section~\ref{sec:correctness} proves that a winning \(\alpha\) satisfies every demand constraint. Conversely, any universal response \(\beta\) satisfying \(F(\alpha,\beta)\) has a compatible index producing an overloaded window.

    \item \emph{Verify the reduction's size and task parameters.}
    Section~\ref{sec:boundedness} proves that the construction takes polynomial time, produces polynomially many tasks with polynomially bounded numerical parameters, and satisfies \(C_i\le D_i\le T_i\) for every task.
\end{enumerate}

The logical correspondence is summarized by the invariant:
\[
\begin{aligned}
&\text{a well-formed selection encoding }\alpha
   \text{ is feasible}\\
&\qquad\Longleftrightarrow\quad
   \text{no assignment }\beta\text{ satisfies }F(\alpha,\beta)\\
&\qquad\Longleftrightarrow\quad
   \forall\beta\;\varphi(\alpha,\beta).
\end{aligned}
\]

\begin{definition}[Local residue patterns and selector locks]
\label{def:local-pattern}
The support of a residue pattern is the set of CRT coordinates whose residues it specifies. A pattern is \emph{local} if its support has size bounded by a constant independent of the source-instance size. The selector lock for \(x_i\) is local to that variable: whether any of its helpers or selectors matches an index \(x\) depends only on the residues at \(t_i,f_i\).
\end{definition}
In this construction, an \(X\)-helper tests two coordinates, a \(Y\)-validity pattern tests one, and a clause pattern tests at most six. Each gadget therefore uses a product of only constantly many coordinate primes, which is the key to the size bound in Section~\ref{sec:boundedness}. All patterns are tested at the same index \(x\); the feasibility argument combines their contributions using the common timing, utilization, and value constraints.

Table~\ref{tab:reduction-notation} collects the notation used throughout the reduction.

\begin{table}[htbp]
\caption{Symbols and notation used in the reduction.}
\label{tab:reduction-notation}
\centering
\small
\renewcommand{\arraystretch}{1.12}
\begin{tabular}{@{}ll@{}}
\toprule
Symbol & Meaning \\
\midrule
\(X,Y\) & existential and universal source variables \\
\(\alpha,\beta\) & assignments to \(X\) and \(Y\) \\
\(\varphi\) & source 3-DNF formula \\
\(F=\neg\varphi=\mathcal C_1\wedge\cdots\wedge \mathcal C_m\) &
counterexample 3-CNF formula \\
\(n,u,m\) & numbers of \(X\)-variables, \(Y\)-variables, and clauses \\
\(A,V(A),K\) & selected task set, its value, and the target \\
\(z,p_z\) & CRT coordinate and its assigned prime \\
\(x\) & integer demand-window index \\
\(\rho,(a_\rho,b_\rho)\) & residue pattern and its compiled congruence \\
\(G(a,b,w)\) & full congruence gadget of weight \(w\) \\
\(H(p,w)\) & single-task max-residue selector gadget \\
\(\Lambda,h\) & block length and early-window threshold \\
\(t_i,f_i\) & true and false CRT coordinates for \(x_i\) \\
\(\tau_i^t,\tau_i^f\) & optional true and false selector tasks \\
\(q,S\) & selector weight and local spike budget \\
\(W_0,V_0\) & anchor weight and forced-task value \\
\bottomrule
\end{tabular}
\end{table}

\section{CRT Demand Gadgets}
\label{sec:gadgets}

The reduction needs a scheduling object that recognizes a conjunction of modular conditions at one time. A coordinate prime stores one logical digit, a residue pattern combines several digits, and the CRT compiles that pattern into one congruence class. A full gadget \(G\) realizes the congruence as exact baseline demand plus a residue-triggered spike. The single-task gadget \(H\) provides the corresponding compact interface for an optional selector: when its max residue matches, its spike is exact; when the residue does not match, its demand is bounded above by the baseline term. This section defines those interfaces, explains the fixed earlier-deadline construction, proves their demand behavior, and identifies the only DBF test-length forms needed later.

\noindent\textbf{Notation.}
A plain \(x\) denotes the integer demand-window index, also called a \emph{CRT timestamp} below. It specifies coordinate residues; the actual
interval length is denoted by \(\ell\). The symbol \(r\) indexes a task position inside a full gadget. Formula variables always carry subscripts, as in \(x_i\) and \(y_j\).

\subsection{Prime Coordinates and Residues}

This subsection defines the CRT representation used to encode universal assignments and to expose local consistency states. Each logical coordinate \(z\) receives a distinct prime \(p_z\). The two residues used for Boolean encodings are
\[
    0\pmod{p_z}
    \qquad\text{and}\qquad
    p_z-1\pmod{p_z}.
\]
An arbitrary timestamp may have another residue; such a timestamp is non-Boolean for that coordinate.
Define the total indicator
\[
\chi_z(x)=
\begin{cases}
1, & x\equiv p_z-1\pmod{p_z},\\
0, & \text{otherwise}.
\end{cases}
\]
Therefore \(\chi_z(x)=1\) says exactly that the max residue \(p_z-1\) is present. The value \(\chi_z(x)=0\) does not by itself encode false: the residue may be \(0\), but it may also be non-Boolean. Exact Boolean encodings will therefore be identified separately.

A finite pattern \(\rho\) specifies exact residues for a set of coordinates \(Z_\rho\). Let
\[
    b_\rho=\prod_{z\in Z_\rho}p_z.
\]
Because the coordinate primes are pairwise coprime, the CRT gives a unique residue \(a_\rho\pmod{b_\rho}\) satisfying all requirements in \(\rho\).
Define
\[
\chi_\rho(x)=
\begin{cases}
1, & x\equiv a_\rho\pmod{b_\rho},\\
0, & \text{otherwise}.
\end{cases}
\]
We say that a timestamp \(x\) \emph{matches} the pattern \(\rho\) if it satisfies every residue requirement in \(\rho\), equivalently if \(\chi_\rho(x)=1\). Matching the coordinate's max residue means \(x\equiv p_z-1\pmod{p_z}\). Failing to match this max residue does not imply that the residue is zero.
For later use, define two further total indicators:
\[
\chi_{a,b}(x)=
\begin{cases}
1, & x\equiv a\pmod b,\\
0, & \text{otherwise},
\end{cases}
\qquad
\chi_p(x)=
\begin{cases}
1, & x\equiv p-1\pmod p,\\
0, & \text{otherwise}.
\end{cases}
\]

\subsection[Full Gadget]{Full Gadget \(G(a,b,w)\)}

Fix integers \(\Lambda,h,w\) with \(2\le h\le\Lambda\) and \(w\ge1\). The concrete values are chosen in Section~\ref{sec:construction}. The common timing scale \(\Lambda\) divides the test lengths into blocks. For each nonnegative integer index \(x\), call \(x\Lambda\) the \emph{baseline test length} and \(x\Lambda+h-1\) the \emph{early test length}. Thus \(h-1\) is the offset of the early test after the preceding block boundary. 
For \(0\le a<b\), the full gadget \(G(a,b,w)\) contains one task for each \(r\in\{1,\ldots,b\}\), with timing parameters
\[ \begin{aligned} C_r&=w,\\ D_r&=\begin{cases} a\Lambda+h-1, & r=a+1,\\ r\Lambda, & r\ne a+1, \end{cases}\\ T_r&=b\Lambda. \end{aligned} \]
Task values are assigned in Section~\ref{sec:construction}. The component with \(r=a+1\) is the \emph{special task}. For comparison, consider a \emph{reference specification} with the same execution times and periods and deadline \(r\Lambda\) for every component. The constructed specification instead gives the special task the fixed relative deadline \(a\Lambda+h-1\), earlier than its reference deadline \((a+1)\Lambda\) by
\[ \delta=(a+1)\Lambda-(a\Lambda+h-1)=\Lambda-h+1. \]
Here \(h-1\) is the offset after the preceding boundary \(a\Lambda\), whereas \(\delta\) is the difference between the two deadlines. This weighted and rescaled form of the congruence-class construction of \citet{ekberg2015constrained} produces one extra weight-\(w\) contribution at early windows whose index satisfies \(x\equiv a\pmod b\). The parameter verification later establishes \(C_r\le D_r\le T_r\) for all instantiated tasks.

\subsection[Max-Residue Gadget]{Max-Residue Gadget \(H(p,w)\)}

An optional selector must be one task, because the target problem permits selection of individual tasks rather than indivisible groups. For a prime \(p\), define the single-task gadget \(H(p,w)\) by its timing tuple \[ (C,D,T)=\bigl(wp,(p-1)\Lambda+h-1,p\Lambda\bigr). \]
Its value is assigned in Section~\ref{sec:construction}. We say that \(x\) matches \(H(p,w)\) when \(x\equiv p-1\pmod p\). At that residue its demand at the early window is exactly \(wx+w\); otherwise its demand is at most \(wx\), as proved below.

\subsection{Gadget Timeline Intuition}

We compare the reference and constructed specifications to explain the extra demand at an early test length. The DBF is evaluated on interval lengths; the timelines below show the synchronous minimum-separation release pattern that realizes each specification's demand count.

Both specifications have \(b\) tasks with execution time \(w\) and period \(b\Lambda\). The reference deadlines are \(\Lambda,2\Lambda,\ldots,b\Lambda\), so the repeated deadline events contribute one weight-\(w\) job at every positive multiple of \(\Lambda\). The constructed specification has the same deadlines except that component \(a+1\) has relative deadline \(a\Lambda+h-1\). Each relative deadline is fixed before execution and applies to every job of its task.

For example, take \(b=5\), \(a=2\), \(\Lambda=10\), \(h=4\), and \(w=1\). The reference deadlines in the first period are \(10,20,30,40,50\); the constructed deadlines are \(10,20,23,40,50\). At length \(20\), both specifications demand two units. At length \(23\), the reference specification still demands two, but the constructed specification demands three. The early test extends the baseline by \(h-1=3\), while the special task's constructed deadline is earlier than its reference deadline by \(\Lambda-h+1=7\).

In general, the special task has a deadline event in \((x\Lambda,x\Lambda+h-1]\) exactly when \(x\equiv a\pmod b\). No non-special deadline lies in this interval because \(0<h-1<\Lambda\). Figure~\ref{fig:full-gadget-timeline} illustrates the resulting extra demand in the matching case. The one-task gadget \(H(p,w)\) uses the same timing principle to detect its max residue, but its nonmatching baseline is only an upper bound.

\begin{figure}[ht]
\centering
\includegraphics[width=0.95\linewidth]{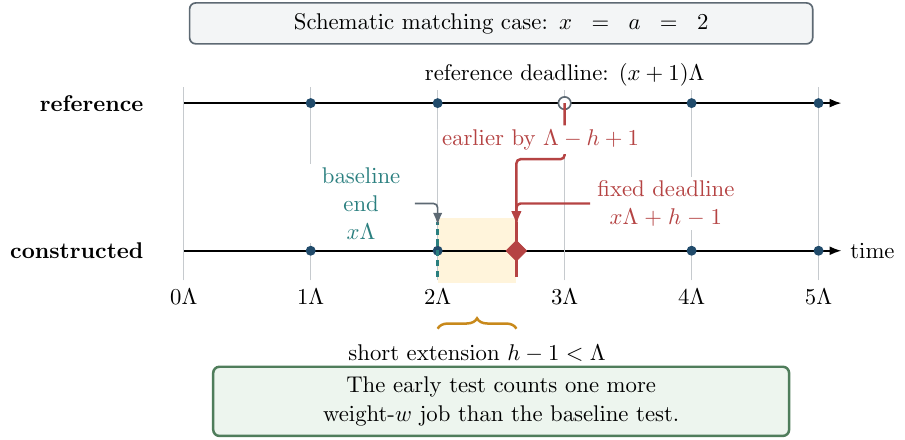}
\caption{Schematic comparison of reference and constructed deadline events when \(x\equiv a\pmod b\). The special job's deadline is \((x+1)\Lambda\) under the reference specification and \(x\Lambda+h-1\) under the constructed specification, earlier by \(\Lambda-h+1\). Both specifications count the same jobs at the baseline length \(x\Lambda\); at the early length \(x\Lambda+h-1\), the constructed specification counts one additional weight-\(w\) job.}
\label{fig:full-gadget-timeline}
\end{figure}

\subsection{Gadget Demand Behavior}

\begin{lemma}
\label{lem:gadget-demand}
Both \(G(a,b,w)\) and \(H(p,w)\) have utilization \(w/\Lambda\). For every integer \(x\ge0\),
\[ \dbf(G(a,b,w),x\Lambda)=wx \]
and
\[ \dbf(G(a,b,w),x\Lambda+h-1) =wx+w\chi_{a,b}(x). \]
Additionally,
\[ \dbf(H(p,w),x\Lambda)\le wx \]
and
\[ \dbf(H(p,w),x\Lambda+h-1) \le wx+w\chi_p(x). \]
The final inequality is an equality whenever \(x\equiv p-1\pmod p\).
\end{lemma}

\begin{proof}
    The full gadget contains \(b\) tasks, each with utilization \(w/(b\Lambda)\), while the max-residue gadget has utilization \(wp/(p\Lambda)\). Therefore each gadget has total utilization \(w/\Lambda\).

    For \(G(a,b,w)\), all \(b\) tasks have period \(b\Lambda\). The reference specification has deadline \(r\Lambda\) for component \(r\), while the constructed specification gives the special component \(a+1\) the fixed earlier deadline \(a\Lambda+h-1\). Because \(0<h-1<\Lambda\), a special job is counted at a baseline length \(x\Lambda\) under the constructed specification exactly when it is counted there under the reference specification. The other deadlines coincide in the two specifications. The reference specification counts exactly one weight-\(w\) job for each integer position \(1,\ldots,x\), so \(\dbf(G(a,b,w),x\Lambda)=wx\).

    Extending the baseline length to \(x\Lambda+h-1\) counts no additional job of a non-special task because the extension is shorter than \(\Lambda\). The special job with release index \(k\ge0\) is counted at the early length exactly when \[ a\Lambda+h-1+kb\Lambda \le x\Lambda+h-1, \] or equivalently \(a+kb\le x\). Compared with the baseline condition \(a+1+kb\le x\), the early test therefore counts one additional job exactly when \(x\equiv a\pmod b\).

    For \(H(p,w)\), the synchronous minimum-separation release pattern has deadline events \((p-1)\Lambda+h-1+kp\Lambda\) for integers \(k\ge 0\). At the early window \(x\Lambda+h-1\), the job with index \(k\) is counted exactly when \(p-1+kp\le x\), so the count is \(\lfloor (x+1)/p\rfloor\). If \(x\equiv p-1\pmod p\), then \(p\) divides \(x+1\) and the demand is \[ wp\cdot\frac{x+1}{p}=w(x+1)=wx+w. \] Otherwise, writing \(s=x\bmod p<p-1\), the count is \((x-s)/p\), so the demand is \(w(x-s)\le wx\). At the baseline window \(x\Lambda\), job \(k\) is counted exactly when \(p-1+kp\le x-1\), because \(0<h-1<\Lambda\); the count is therefore \(\lfloor x/p\rfloor\), and the demand is \(wp\lfloor x/p\rfloor\le wx\).
\end{proof}

\noindent\textbf{Spike accounting.}
A full gadget \(G(a,b,w)\) matches \(x\) when \(\chi_{a,b}(x)=1\). Its demand beyond its exact baseline \(wx\) is \(w\chi_{a,b}(x)\). For \(H(p,w)\), the quantity \(w\chi_p(x)\) is a \emph{spike allowance}: when \(x\) matches its max residue, demand is exactly \(wx+w\); otherwise demand may be smaller than \(wx\), so zero is an upper bound on the excess rather than an exact contribution.

Upper-bound arguments may sum these allowances. Lower-bound arguments use the CRT to choose a timestamp at which every selected selector matches, making its bound exact. If the selected gadgets have total weight \(\Lambda\), then
\[ \dbf(A,x\Lambda+h-1) \le x\Lambda+\text{(sum of matching spike allowances)}. \]
Only \(h-1\) units fit beyond \(x\Lambda\). At a timestamp where all required bounds are attained, an aggregate spike of \(h\) produces demand \(\ell+1\) at \(\ell=x\Lambda+h-1\), violating feasibility.

\subsection{Relevant DBF Test Lengths}

\begin{lemma}[Relevant DBF test lengths]
\label{lem:jump-forms}
For every task with timing parameters \((C,D,T)\) generated by \(G(a,b,w)\) or \(H(p,w)\), every length \(D+kT\), for an integer \(k\ge0\), has the form \(x\Lambda\) or \(x\Lambda+h-1\) for a nonnegative integer \(x\).
\end{lemma}
\begin{proof}
For a full gadget, \(D\) is either \(r\Lambda\) or \(a\Lambda+h-1\), and \(T=b\Lambda\). For a max-residue gadget, \(D=(p-1)\Lambda+h-1\) and \(T=p\Lambda\). Adding \(kT\) preserves the claimed form in both cases.
\end{proof}
Together with Lemma~\ref{lem:jump}, this proves that the two displayed families of lengths contain every DBF test needed for feasibility.

The gadgets now provide the temporal part of the reduction: a timestamp determines which residue patterns match, and the matching patterns determine the early-window spikes. What the gadgets do not yet guarantee is that a selected set represents one valid existential assignment. The next section provides that local consistency condition.

\section{Local Selector Lock with Polynomial Weights}
\label{sec:selector-lock}

The construction must force every feasible selection \(A\) meeting the value target to choose exactly one truth selector for each existential variable. Recall that meeting the target means \(V(A)\ge K\); feasibility is a separate requirement. Value and utilization first force exactly \(n\) optional selectors in total, but one variable could still receive both selectors while another receives neither.

The selector lock is local to \(x_i\) in the sense of Definition~\ref{def:local-pattern}: its helpers and selectors test only the two coordinates \(t_i,f_i\). At an exact both-match state, a duplicated choice produces excess demand. At an exact neither-match state, the helpers for a missing choice fill its local budget and provide no compensating slack. CRT realizes these per-variable states at one common index. This section proves the per-variable demand inequalities; Lemma~\ref{lem:lock} combines them after Section~\ref{sec:construction} fixes the global timing and value parameters.

\subsection{Selection Choices and Timestamp Tests}

For each existential variable \(x_i\), create two CRT coordinates \[ t_i,\qquad f_i, \] with distinct primes \(p_{t_i}\) and \(p_{f_i}\). These coordinates belong to the timestamp encoding; they are not tasks and are not selection choices.

The selection choices are two optional selector tasks. For a common weight \(q\), define
\[
    \tau_i^t=H(p_{t_i},q),
    \qquad
    \tau_i^f=H(p_{f_i},q).
\]
Selecting \(\tau_i^t\) encodes \(x_i=\mathrm{true}\), while selecting \(\tau_i^f\) encodes \(x_i=\mathrm{false}\). The symbols \(\tau_i^t,\tau_i^f\) denote tasks; \(t_i,f_i\) denote the CRT coordinates inspected by those tasks.

The timestamp tests those claims. At an early window indexed by \(x\), \(\tau_i^t\) has an exact spike of \(q\) when \(x\equiv p_{t_i}-1\pmod{p_{t_i}}\), and \(\tau_i^f\) has an exact spike of \(q\) when \(x\equiv p_{f_i}-1\pmod{p_{f_i}}\). A selector that is not selected contributes no demand, regardless of the timestamp residue.

\subsection{Exact Local Timestamp States}

Define
\[
    \sigma_i(x)=
    \bigl(\chi_{t_i}(x),\chi_{f_i}(x)\bigr).
\]
The four indicator states are
\[
    (1,0),\qquad(0,1),\qquad(1,1),\qquad(0,0).
\]
These are timestamp states, not selection cases. We call \(x\) \emph{exact for \(x_i\)} when both residues are Boolean:
\[
x\bmod p_{t_i}\in\{0,p_{t_i}-1\},
\qquad
x\bmod p_{f_i}\in\{0,p_{f_i}-1\}.
\]
At an exact timestamp, \((1,0)\) and \((0,1)\) encode true and false, while \((1,1)\) and \((0,0)\) are the both-match and neither-match test states. These states describe the timestamp and are not four selection choices. In particular, \((0,0)\) is not an optional \enquote{choose neither} gadget. The mandatory \((0,0)\)-helper introduced below supplies the full local spike budget \(S\) at an exact neither-match timestamp. Lemmas~\ref{lem:missing} and~\ref{lem:lock} use this contribution to ensure that a variable with no selected selector creates no slack that could offset the excess demand produced by a variable with both selectors. At a non-exact timestamp, an indicator value of zero does not necessarily mean that the residue is zero. This is why the exactness condition is stated separately.

\subsection{Selector Weight and Local Budget}

Recall that \(m\) is the number of clauses in \(F=\mathcal C_1\wedge\cdots\wedge \mathcal C_m\). The entire clause layer can change the early-window spike by at most \(m\). Set \[ q=m+1, \qquad S=q+1=m+2. \]
Therefore one selected selector of weight \(q\), together with a unit helper, exactly fills the intended local budget \(S\). A duplicate selector creates an excess of \(q=m+1\), which cannot be hidden by losing at most \(m\) units from the clause layer.

\subsection{Four Exact-State Helper Gadgets}

Let \[ b_i=p_{t_i}p_{f_i}. \] For \(s\in\{10,01,11,00\}\), let \(a_i^s\pmod{b_i}\) be the CRT residue that realizes the corresponding exact pair of Boolean residues. Add the four full helper gadgets
\[
\begin{array}{ccl}
G(a_i^{10},b_i,1), &\qquad& G(a_i^{01},b_i,1),\\
G(a_i^{11},b_i,1), && G(a_i^{00},b_i,S).
\end{array}
\]
Each expression denotes a full gadget and therefore expands into \(b_i\) ordinary sporadic tasks. These helper tasks are not optional selectors. Section~\ref{sec:construction} assigns every component task the high value \(V_0\), and Lemma~\ref{lem:mandatory} proves that every selection meeting the value target must include all four helpers. Therefore \enquote{forced helper} describes a property proved later, not an assumption made here.

At an exact timestamp, exactly one of the four helper patterns matches. The matching helper contributes a spike of \(1\) in states \((1,0)\), \((0,1)\), and \((1,1)\), and a spike of \(S\) in state \((0,0)\). At a non-exact timestamp, none of the four exact-state helpers matches.

The intended lock is summarized in Figure~\ref{fig:selector-lock}. A valid one-selector selection stays within the local budget \(S\) at every timestamp. If both selectors are selected, the exact state \((1,1)\) activates both selector spikes and produces a local spike of \(S+q\), exceeding the budget by \(q\). If neither selector is selected, the exact state \((0,0)\) still contributes the full \(S\), so the missing variable creates no slack that could absorb the duplicate's excess. Therefore \((1,0)\), \((0,1)\), \((1,1)\), and \((0,0)\) are timestamp tests for valid, duplicate, and missing selections, and not four selection tests.

\begin{figure}[!ht]
\centering
\includegraphics[width=0.94\linewidth]{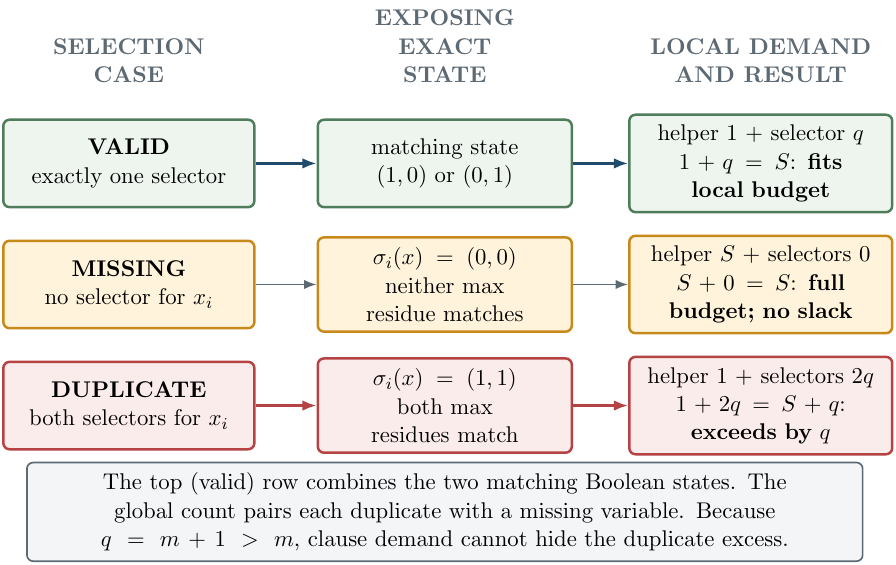}
\caption{The local selector lock at exposing exact timestamp states. A valid one-selector choice stays within its local budget \(S=q+1\). A missing choice attains the full budget \(S\) at an exact \((0,0)\) timestamp, leaving no compensating slack there. When the total selector count is \(n\), every duplicated choice is paired with a missing choice elsewhere. At an exact \((1,1)\) timestamp, a duplicate contributes \(S+q\). The CRT realizes all these local states simultaneously, and \(q=m+1\) exceeds the entire clause layer's possible variation.}
\label{fig:selector-lock}
\end{figure}
\FloatBarrier

\subsection{Local Spike Allowance and Budget Margin}

For a selection \(A\) that contains all four exact-state helper gadgets for \(x_i\), let \(\mathcal L_i(A)\) denote those helpers together with whichever of \(\tau_i^t,\tau_i^f\) are selected. Define
\[
    s_i(A)=
    \mathbf 1[\tau_i^t\in A]
    +\mathbf 1[\tau_i^f\in A].
\]
The local early-window spike allowance is
\[
\begin{aligned}
L_i(A,x)={}&
    \chi_{a_i^{10},b_i}(x)
    +\chi_{a_i^{01},b_i}(x)
    +\chi_{a_i^{11},b_i}(x)
    +S\chi_{a_i^{00},b_i}(x)\\
    &+q\,\mathbf 1[\tau_i^t\in A]\chi_{t_i}(x)
    +q\,\mathbf 1[\tau_i^f\in A]\chi_{f_i}(x).
\end{aligned}
\]
The four helpers have total baseline weight \(S+3\), and the selected selectors have total baseline weight \(q\,s_i(A)\). Additivity of demand and Lemma~\ref{lem:gadget-demand} therefore give
\[
    \dbf\bigl(\mathcal L_i(A),x\Lambda+h-1\bigr)
    \le
    x\bigl(S+3+q\,s_i(A)\bigr)+L_i(A,x).
\]
Every helper term in \(L_i(A,x)\) is an exact spike contribution. A selector term is an upper allowance in general and is exact when that
selected selector matches its max residue. Therefore the inequality is an equality at each timestamp used below to establish a lower bound.

Define the signed local budget margin \[ \mu_i(A,x)=S-L_i(A,x). \] A positive value is unused spike allowance, zero means that the allowance exactly fills the local budget, and a negative value is an overage. The next three lemmas bound this margin for valid selections at every timestamp and compute it at the exact states used to expose duplicate and missing selections.

\subsection{Valid-Selection Lemma}

\begin{lemma}[Valid local selection]
\label{lem:valid}
Suppose that \(A\) contains all four helpers for \(x_i\) and exactly one of \(\tau_i^t,\tau_i^f\). Then, for every timestamp \(x\),
\[
    L_i(A,x)\le S
    \qquad\text{and hence}\qquad
    \mu_i(A,x)\ge 0.
\]
As a result, \[ \dbf\bigl(\mathcal L_i(A),x\Lambda+h-1\bigr) \le x(S+3+q)+S. \]
\end{lemma}

\begin{proof}
Assume first that only \(\tau_i^t\) is selected. If \(x\) is exact for \(x_i\), exactly one helper pattern matches, and Table~\ref{tab:valid-selector} lists all possibilities.
\begin{table}[tbp]
\caption{Local spike allowances when only \(\tau_i^t\) is selected. The selector column is an upper allowance; it is an exact excess when that selector matches.}
\label{tab:valid-selector}
\centering
\begin{tabular}{ccccc}
\toprule
State \(\sigma_i(x)\) & Helper & \(\tau_i^t\) & \(L_i(A,x)\) & \(\mu_i(A,x)\)\\
\midrule
\((1,0)\) & \(1\) & \(q\) & \(S\) & \(0\)\\
\((0,1)\) & \(1\) & \(0\) & \(1\) & \(q\)\\
\((1,1)\) & \(1\) & \(q\) & \(S\) & \(0\)\\
\((0,0)\) & \(S\) & \(0\) & \(S\) & \(0\)\\
\bottomrule
\end{tabular}
\end{table}
Therefore \(L_i(A,x)\le S\) and \(\mu_i(A,x)\ge0\) in every exact state.

If \(x\) is not exact for \(x_i\), none of the four helper patterns matches. Therefore \[ L_i(A,x)=q\chi_{t_i}(x)\in\{0,q\}, \] so \(L_i(A,x)\le q<S\) and \(\mu_i(A,x)\ge S-q=1\).

If only \(\tau_i^f\) is selected, interchange the \(t_i\)- and \(f_i\)-coordinates. At exact timestamps its selector allowance is \(q\) in states \((0,1)\) and \((1,1)\), and zero in states \((1,0)\) and \((0,0)\); at non-exact timestamps, \(L_i(A,x)=q\chi_{f_i}(x)\le q<S\). Therefore the same two inequalities hold for every \(x\). The displayed DBF bound now follows from \(s_i(A)=1\) and the local-demand inequality above.
\end{proof}

\noindent\textbf{Equality cases used later.}
If only \(\tau_i^t\) is selected, exact states \((1,0)\), \((1,1)\), and \((0,0)\) have \(L_i(A,x)=S\). If only \(\tau_i^f\) is selected, exact states \((0,1)\), \((1,1)\), and \((0,0)\) have \(L_i(A,x)=S\). At the matching states \((1,0)\) and \((0,1)\), respectively, the selected selector is active, so the local DBF upper bound is attained.

\subsection{Duplicate-Selection Lemma}

\begin{lemma}[Duplicate local selection]
\label{lem:duplicate}
Suppose that \(A\) contains all four helpers and both selectors for \(x_i\). At every exact timestamp \(x\) with \(\sigma_i(x)=(1,1)\),
\[ L_i(A,x)=S+q, \qquad \mu_i(A,x)=-q. \] Also, the local demand is exactly \[ \dbf\bigl(\mathcal L_i(A),x\Lambda+h-1\bigr) =x(S+3+2q)+(S+q). \]
Therefore a duplicate creates an actual overage of \(q=m+1\) above the local spike budget.
\end{lemma}

\begin{proof}
At exact state \((1,1)\), the unique matching helper contributes \(1\), and both selected selectors match their max residues. Thus \(L_i(A,x)=1+2q=S+q\) and \(\mu_i(A,x)=-q\). By Lemma~\ref{lem:gadget-demand}, the helper demand is exactly \(x(S+3)+1\), and each selector contributes \(qx+q\). Summing gives the claimed DBF identity.
\end{proof}

\subsection{Missing-Selection Lemma}

\begin{lemma}[Missing local selection]
\label{lem:missing}
Suppose that \(A\) contains all four helpers and neither selector for \(x_i\). At every exact timestamp \(x\) with \(\sigma_i(x)=(0,0)\),
\[ L_i(A,x)=S, \qquad \mu_i(A,x)=0. \] Additionally, the local demand is exactly \[ \dbf\bigl(\mathcal L_i(A),x\Lambda+h-1\bigr) =x(S+3)+S. \]
Therefore a missing selector contributes no local overage, but it provides exactly zero local margin with which to offset a duplicate elsewhere.
\end{lemma}

\begin{proof}
At exact state \((0,0)\), only the weight-\(S\) helper matches. Since no selector is selected, \(L_i(A,x)=S\) and
\(\mu_i(A,x)=0\). Lemma~\ref{lem:gadget-demand} gives the exact helper demand \(x(S+3)+S\), proving the claimed DBF identity.
\end{proof}

\subsection{Two-Variable Selector Example}

Let \(n=2\) and \(m=2\), so \(q=3\) and \(S=4\). Consider first the well-formed selection \[ \Gamma_{\mathrm{mand}}\uplus\{\tau_1^t,\tau_2^f\}. \] At an exact timestamp placing \(x_1\) in state \((1,0)\) and \(x_2\) in state \((0,1)\), each variable contributes its full local budget \(S\). The \(X\)-layer spike is therefore
\[ S+S=8. \]

Now consider the malformed selection \[ A=\Gamma_{\mathrm{mand}}\uplus\{\tau_1^t,\tau_1^f\}, \] which selects both selectors for \(x_1\) and none for \(x_2\). It still contains exactly \(n=2\) selectors, so the global selector count alone does not reject it. By the CRT, choose one timestamp placing \(x_1\) in exact state \((1,1)\) and \(x_2\) in exact state \((0,0)\). The local spikes are \[ x_1:\ 1+q+q=7, \qquad x_2:\ S=4. \] Therefore the \(X\)-layer spike is \[ 7+4=11=2S+q, \] which exceeds its intended budget \(2S=8\) by \(q=3\).

This local calculation becomes a DBF violation after the global parameters are fixed. Section~\ref{sec:construction} sets \[ h=W_0+nS+u+m. \] In the proof of Lemma~\ref{lem:lock}, the CRT timestamp is also chosen so that the anchor contributes \(W_0\), the \(Y\)-validity layer contributes \(u\), and every selected selector needed by the malformed selection matches. Even if the clause layer contributes nothing, the total spike is at least
\[
    W_0+nS+q+u
    =W_0+nS+u+m+1
    =h+1.
\]
At the early window \(\ell=x\Lambda+h-1\), Lemma~\ref{lem:gadget-demand} therefore gives \[ \dbf(A,\ell)\ge x\Lambda+h+1=\ell+2, \] which is a deadline miss. This is the precise connection between the local selector calculation, the generated tasks, and the global DBF inequality.

With the temporal gadgets and local selector lock in place, the reduction can now be assembled layer by layer. The construction below chooses the global weights so that the local budget \(S\), the logical clause contribution, and the one-unit feasibility gap share one common early-window threshold.

\section{Reduction Construction}
\label{sec:construction}

We now specify the residue patterns and assemble the task inventory summarized in Section~\ref{sec:construction-preview}.

\subsection{Residue Patterns}

Assign the first \(2n+u\) primes to the coordinates \(t_i,f_i,y_j\). For each pattern, compute its canonical CRT residue
\(0\le a_\rho<b_\rho\), where \(b_\rho\) is the product of its coordinate primes. These computations take polynomial time; Section~\ref{sec:boundedness} bounds the primes, moduli, and explicit output size.

The anchor \(G(0,1,W_0)\) matches every timestamp. Let \(\mathcal P\) contain the four weighted \(X\)-helper patterns per variable from Section~\ref{sec:selector-lock}, the \(Y\)-validity patterns, and the clause patterns defined below. Patterns form an indexed multiset, so repeated congruences retain their separate demand contributions.

For each \(y_j\), include two unit-weight validity patterns, with residues \(0\) and \(p_{y_j}-1\), encoding false and true. At most one matches at any timestamp, and neither matches a non-Boolean residue.

For each clause \(\mathcal C_k\), let \(Z_k\) contain one coordinate \(y_j\) for each universal variable occurring in that clause and both \(t_i,f_i\) for each existential variable occurring in it. Set \(b_k=\prod_{z\in Z_k}p_z\). For each of the clause's seven satisfying assignments \(\gamma\), include the unit-weight pattern \(a_{k,\gamma}\pmod{b_k}\) encoding \(\gamma\). Existential true and false require the exact states \((1,0)\) and \((0,1)\), respectively; universal true and false require residues \(p_{y_j}-1\) and \(0\). The seven patterns are pairwise disjoint. At a valid Boolean encoding, exactly one matches if the clause is satisfied, and none otherwise. Reading both existential coordinates prevents states \((1,1)\) and \((0,0)\) from matching a clause pattern that mentions that variable.

\subsection{Assembly}

Instantiate each \(\rho\in\mathcal P\) as \(G(a_\rho,b_\rho,w_\rho)\), include the anchor and optional selectors, and use the timing constants, task multisets, and value assignment defined in Section~\ref{sec:construction-preview}. The mandatory pattern weights sum to \(B_{\mathrm{mand}}=\sum_{\rho\in\mathcal P}w_\rho=n(S+3)+2u+7m\), excluding the anchor. The individual task counts are \(M=1+\sum_{\rho\in\mathcal P}b_\rho\), \(|\Gamma_{\mathrm{opt}}|=2n\), and \(|\Gamma|=M+2n\); \(m\) counts clauses. These patterns therefore complete the instance \((\Gamma,K)\) specified in the overview. Figure~\ref{fig:early-spike-budget} summarizes the early-window gap used in the correctness proof.

\begin{figure}[!ht]
\centering
\includegraphics[width=0.97\linewidth]{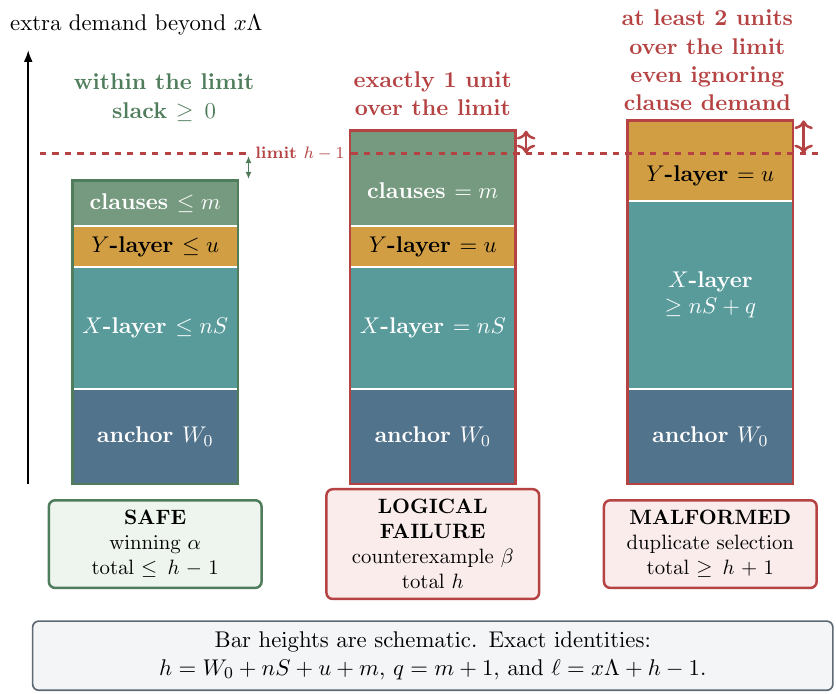}
\caption{Excess demand above \(x\Lambda\) at the early test length \(\ell=x\Lambda+h-1\). A selection encoding a winning assignment stays at or below \(h-1\) at every timestamp. For a well-formed selection encoding \(\alpha\), a matching timestamp encoding a counterexample \(\beta\) with \(F(\alpha,\beta)=\mathrm{true}\) realizes exact excess \(h\), producing a one-unit overload. For a malformed selection containing all mandatory tasks and exactly \(n\) selectors, CRT realizes excess at least \(h+1\), even when clause demand is ignored.}
\label{fig:early-spike-budget}
\end{figure}

\section{Correctness Proof}
\label{sec:correctness}

The proof has two phases:
\begin{enumerate}
    \item The \emph{structural phase} identifies the only possible shape of a feasible selection meeting the value target. Value forces every mandatory task (Lemma~\ref{lem:mandatory}), value and utilization together force exactly \(n\) selectors in total (Lemma~\ref{lem:squeeze}), and the local lock distributes them one per existential variable (Lemma~\ref{lem:lock}). The selection therefore defines one assignment \(\alpha\) (Lemma~\ref{lem:extract}).
    \item The \emph{logical phase} relates that assignment to the source formula. The forward direction proves that a winning \(\alpha\) gives a feasible selection. The reverse direction turns any losing response \(\beta\) into a violating early window (Lemmas~\ref{lem:forward} and~\ref{lem:reverse}).
\end{enumerate}

As everywhere in the paper, \(A\) denotes a selected task set, while \(S\) is the local budget constant of Section~\ref{sec:selector-lock}.

\subsection{Mandatory Selection Lemma}

\begin{lemma}[Mandatory selection]
\label{lem:mandatory}
Every selection meeting the value target must include all mandatory tasks.
\end{lemma}

\begin{proof}
    There are \(M\) mandatory tasks, each with value \(V_0\), and \(2n\) optional selector tasks, each with value \(1\). If a selection excludes even one mandatory task, its value is at most
    \[
        (M-1)V_0+2n.
    \]
    Since \(V_0>2n\), this is strictly less than
    \[
        MV_0+n=K.
    \]
    Therefore every selection meeting the value target includes all mandatory tasks.
\end{proof}

\subsection{Utilization Squeeze Lemma}

\begin{lemma}[Utilization squeeze]
\label{lem:squeeze}
Every feasible selection meeting the value target contains exactly \(n\) selected optional selectors.
\end{lemma}

\begin{proof}
By Lemma~\ref{lem:mandatory}, every selection meeting the value target contains all mandatory tasks. Since the mandatory tasks contribute exactly \(MV_0\) value, the target \[ K=MV_0+n \] forces the selection to contain at least \(n\) optional selectors.

Let \(s\) be the number of selected optional selectors. The mandatory gadgets contribute utilization \[ \frac{W_0+B_{\mathrm{mand}}}{\Lambda}, \] and each selected optional selector contributes \(q/\Lambda\). Therefore \[ U=\frac{W_0+B_{\mathrm{mand}}+sq}{\Lambda}. \] Since \[ \Lambda=W_0+B_{\mathrm{mand}}+nq, \]
feasibility implies \(U\le1\) and hence \(s\le n\). Combining the value requirement \(s\ge n\) with the utilization requirement \(s\le n\) gives \(s=n\). Therefore every feasible selection meeting the value target contains exactly \(n\) selected optional selectors.
\end{proof}

\subsection{Local Selector-Lock Lemma}
With the timing constants in place, the local lemmas of Section~\ref{sec:selector-lock} combine into the global selector lock.

\begin{lemma}[Local selector lock]
\label{lem:lock}
Every feasible selection meeting the value target contains at most one selected selector for each existential variable. Consequently, for every existential variable \(x_i\), it contains exactly one member of \(\{\tau_i^t,\tau_i^f\}\).
\end{lemma}

\begin{proof}
Consider a feasible selection \(A\) meeting the value target. By Lemma~\ref{lem:squeeze}, \(A\) contains exactly \(n\) selected optional selectors. Suppose for contradiction that both selectors are selected for some variable, and let \(d\ge1\) be the number of variables with both selectors selected. The number of variables with neither selector selected is also \(d\). Relative to the intended one selector per variable, each duplicate contributes one excess selector and each missing variable contributes one deficit.
Use the CRT to choose a timestamp \(x\) satisfying the following exact residue requirements:
\begin{itemize}
    \item every variable with both selectors selected is placed in state \((1,1)\);
    \item every variable with neither selector selected is placed in state \((0,0)\);
    \item every variable with exactly one selected selector is placed in a budget-reaching valid state: \((1,0)\) if \(\tau_i^t\) is selected and \((0,1)\) if \(\tau_i^f\) is selected; and 
    \item every universal coordinate \(y_j\) is placed at the residue \(0\pmod{p_{y_j}}\), which is one of its Boolean residues.
\end{itemize}
Each requirement fixes one residue modulo its coordinate prime, and no coordinate receives two different requirements. Because all coordinate moduli are distinct primes, they are pairwise coprime. The Chinese Remainder Theorem therefore gives a simultaneous solution that is unique modulo the product of the coordinate primes. Choose the least nonnegative representative as \(x\). As noted in Section~\ref{sec:construction}, this representative is computable in polynomial time by the constructive CRT algorithm.

At the chosen timestamp, every variable with exactly one selected selector is in the matching valid state for that selector. The equality cases following Lemma~\ref{lem:valid} therefore give \(L_i(A,x)=S\) for each such variable. Lemma~\ref{lem:missing} gives \(L_i(A,x)=S\) for each missing variable, and Lemma~\ref{lem:duplicate} gives \(L_i(A,x)=S+q\) for each duplicated variable. Because there are \(d\) duplicated variables, \(d\) missing variables, and \(n-2d\) one-selector variables,
\[
    \begin{aligned}
    \sum_{i=1}^{n}L_i(A,x)
    &=(n-2d)S+dS+d(S+q)\\
    &=nS+dq\\
    &\ge nS+q.
    \end{aligned}
\]
These are actual spike contributions, not just upper allowances: every selected selector matches its max residue at the chosen timestamp, and every selected helper state is exact.

The \(Y\)-validity layer contributes its exact spike \(u\), because every universal coordinate is at a Boolean residue. The dummy anchor contributes its exact spike \(W_0\), and the clause-layer spike is nonnegative. Moreover, every mandatory full gadget contributes its exact baseline demand and the \(n\) selected selectors contribute total baseline demand \(xnq\). Therefore Lemma~\ref{lem:gadget-demand} gives
\[
    \begin{aligned}
    \dbf(A,x\Lambda+h-1)
    &\ge
    x(W_0+B_{\mathrm{mand}}+nq)
    +W_0+nS+dq+u\\
    &\ge
    x\Lambda+W_0+nS+q+u\\
    &=
    x\Lambda+h+1\\
    &>
    x\Lambda+h-1,
    \end{aligned}
\]
where the equality uses \(q=m+1\) and \(h=W_0+nS+u+m\). This contradicts feasibility, so no variable has both selectors selected. Finally, \(A\) contains exactly \(n\) selectors by Lemma~\ref{lem:squeeze}; because there are \(n\) variables and each has at most one selector, every variable has exactly one selector.
\end{proof}

\subsection{Assignment Extraction Lemma}

\begin{lemma}[Assignment extraction]
\label{lem:extract}
Every feasible selection meeting the value target defines a unique assignment to the existential variables.
\end{lemma}

\begin{proof}
Let \(A\) be a feasible selection meeting the value target. By Lemma~\ref{lem:lock}, exactly one member of \(\{\tau_i^t,\tau_i^f\}\) is selected for every \(x_i\). Define
\[
    \alpha(x_i)=
    \begin{cases}
        \mathrm{true}, & \tau_i^t \text{ is selected},\\
        \mathrm{false}, & \tau_i^f \text{ is selected}.
    \end{cases}
\]
Exactly one case applies to each \(x_i\), so \(\alpha\) is a well-defined and unique assignment.
\end{proof}

The structural phase is now complete: a feasible selection meeting the value target is not just a set of the right cardinality, but a unique assignment \(\alpha\). The remaining question is logical. The forward direction shows that a winning \(\alpha\) survives every early-window test. The reverse direction converts any losing response \(\beta\) into one explicit violating window.

\subsection{Forward Direction}

\begin{lemma}[Forward direction]
\label{lem:forward}
If the source formula \(\exists X\forall Y\,\varphi(X,Y)\) is true, then the constructed \(\VTS\) instance has a feasible selected subset meeting the value target \(K\).
\end{lemma}

\begin{proof}
    Let \(\alpha\) be a winning assignment for \(X\), so \(\varphi(\alpha,\beta)\) is true for every assignment \(\beta\) to \(Y\). Equivalently, \(F(\alpha,\beta)\) is false for every \(\beta\).

    \noindent\textbf{Construct the selection.}
    Construct a selection \(A_\alpha\) as follows. Select every mandatory task. For each existential variable \(x_i\), select \(\tau_i^t\) if \(\alpha(x_i)=\mathrm{true}\), and select \(\tau_i^f\) otherwise. The value is exactly \(MV_0+n=K\), and the utilization is exactly
    \[
        \frac{W_0+B_{\mathrm{mand}}+nq}{\Lambda}=1.
    \]

    \noindent\textbf{Baseline windows.}
    The remaining step is to show that the demand never exceeds the length of any relevant window. By Lemmas~\ref{lem:jump} and~\ref{lem:jump-forms}, it is enough to check windows of the form \(x\Lambda\) and \(x\Lambda+h-1\). At \(x\Lambda\), the full gadgets contribute exactly \[ x(W_0+B_{\mathrm{mand}}), \] and the \(n\) selected optional max-residue gadgets contribute at most \[ xnq. \] Therefore \[ \dbf(A_\alpha,x\Lambda)\le x\Lambda. \]
          
    \noindent\textbf{Early windows.}
    Now consider an early window \(x\Lambda+h-1\), and let \[ \Delta_\alpha(x)=\dbf(A_\alpha,x\Lambda+h-1)-x\Lambda \] be the actual demand beyond the baseline budget \(x\Lambda\). By Lemma~\ref{lem:gadget-demand} and the spike-accounting convention, this excess is at most the sum of the matching spike allowances, as summarized in Figure~\ref{fig:early-spike-budget}. The dummy anchor contributes \(W_0\). By Lemma~\ref{lem:valid}, the total \(X\)-layer spike allowance
    is at most \(nS\). The \(Y\)-validity layer contributes at most \(u\), since at most one of the two validity patterns for each \(y_j\) can match.
    The clause layer contributes at most \(m\), since at most one satisfying pattern for each clause can match. Therefore \[ \Delta_\alpha(x)\le W_0+nS+u+m=h. \] Both \(\dbf(A_\alpha,x\Lambda+h-1)\) and \(x\Lambda\) are integers, so \(\Delta_\alpha(x)\) is an integer. As a result, this bound leaves only the safe case \(\Delta_\alpha(x)\le h-1\) and the equality case \(\Delta_\alpha(x)=h\). If \(\Delta_\alpha(x)\le h-1\), then \[ \dbf(A_\alpha,x\Lambda+h-1)\le x\Lambda+h-1,
    \] as required.
                
    \noindent\textbf{Ruling out \(\Delta_\alpha(x)=h\).}
    Suppose that \(\Delta_\alpha(x)=h\). The preceding upper bound was obtained by adding the anchor contribution \(W_0\), the \(X\)-layer bound \(nS\) from Lemma~\ref{lem:valid}, the \(Y\)-validity bound \(u\), and the clause-layer bound \(m\) from Section~\ref{sec:construction}. Because these four bounds sum to \(h\), equality can hold only if every individual layer bound is tight.
    As a result:
    \begin{itemize}
        \item every \(y_j\) is at a Boolean residue, defining an assignment
                    \(\beta\);
        \item every clause of \(F\) has one matched satisfying pattern; and
        \item every \(x_i\) attains its full local spike allowance \(S\).
    \end{itemize}
    Consider any matched clause pattern that mentions \(x_i\). It fixes the exact pair \(\sigma_i(x)=(1,0)\) or \((0,1)\). If the fixed pair disagreed with the selected selector, say the pattern fixes \((0,1)\) while \(\tau_i^t\) is selected, then at this timestamp the helper contributes \(1\) and the selector does not spike, so \(x_i\) contributes \(1<S\), contradicting tightness of the \(X\)-layer. Therefore every matched clause pattern agrees with \(\alpha\) on the existential variables it mentions and with \(\beta\) on the universal variables it mentions, and the matched patterns certify that
    \[
      F(\alpha,\beta)=\mathrm{true}.
    \]
    This contradicts the assumption that \(\alpha\) is winning. Therefore \(\Delta_\alpha(x)\ne h\). Since \(\Delta_\alpha(x)\) is an integer and the preceding bound gives \(\Delta_\alpha(x)\le h\), it follows that \(\Delta_\alpha(x)\le h-1\). Therefore
    \[
        \dbf(A_\alpha,x\Lambda+h-1)
        \le x\Lambda+h-1
    \]
    at every early window.
\end{proof}
                          
\subsection{Reverse Direction}
                          
\begin{lemma}[Reverse direction]
\label{lem:reverse}
    If the constructed \(\VTS\) instance has a feasible selected subset meeting the value target \(K\), then the source formula \(\exists X\forall Y\,\varphi(X,Y)\) is true.
\end{lemma}
                          
\begin{proof}
    Suppose a feasible selected subset \(A\) meets the value target \(K\). By Lemmas~\ref{lem:mandatory},~\ref{lem:squeeze}, and~\ref{lem:extract}, \(A\) selects all mandatory tasks and exactly one selector for every existential variable. Let \(\alpha\) be the extracted assignment to \(X\).
                              
    \noindent\textbf{From a losing existential assignment to a DBF violation.}
    Assume for contradiction that \(\alpha\) is not winning. Then some assignment \(\beta\) to \(Y\) satisfies \[ F(\alpha,\beta)=\mathrm{true}. \] For each \(x_i\), require both coordinate residues to be Boolean and impose \(\sigma_i(x)=(1,0)\) when \(\alpha(x_i)=\mathrm{true}\) and \(\sigma_i(x)=(0,1)\) otherwise. For each \(y_j\), impose the Boolean residue encoding \(\beta(y_j)\). These conditions fix one residue modulo each distinct coordinate prime, so the CRT gives a timestamp \(x\) satisfying all of them simultaneously.
                                
    Consider the early window \[ \ell=x\Lambda+h-1. \] The dummy anchor contributes its exact spike \(W_0\). Every selected optional selector matches its max residue. By the equality cases following Lemma~\ref{lem:valid}, each existential variable's matching helper and selected selector attain an actual local spike of \(S\). The \(X\)-layer spike is therefore exactly \(nS\).

    By the \(Y\)-validity construction in Section~\ref{sec:construction}, exactly one validity gadget matches for each encoded \(y_j\), giving a total spike of \(u\). Because \((\alpha,\beta)\) satisfies every clause of \(F\), the clause construction in Section~\ref{sec:construction} makes exactly one satisfying pattern match for each clause, giving a total clause spike of \(m\). Therefore the actual spike beyond the baseline \(x\Lambda\) is \[ W_0+nS+u+m=h. \] All selected selectors match at this timestamp, so the selector bounds in Lemma~\ref{lem:gadget-demand} are equalities. That lemma therefore gives
    \[ \dbf(A,\ell) = x\Lambda+h = \ell+1 > \ell. \] This is a DBF violation and contradicts feasibility. Hence no losing response \(\beta\) exists, and \(\alpha\) is winning.
\end{proof}

This completes the logical part of the reduction. The remaining task is to verify that the CRT construction does not cause an exponential increase in the numerical parameters or number of tasks.
                                     
\section{Polynomial Boundedness and Unary Encoding}
\label{sec:boundedness}

Global CRT timestamps may have exponentially large magnitude, but they are not task parameters. The following lemma bounds the explicit task expansion and all generated integers.

\begin{lemma}[Polynomial construction]
\label{lem:bounded-construction}
The construction takes polynomial time and produces polynomially many tasks satisfying \(C_i\le D_i\le T_i\). Every generated \(C_i,D_i,T_i,V_i,K\) has magnitude polynomial in the source encoding length, so the construction is polynomial under binary or unary output encoding.
\end{lemma}

\begin{proof}
Put \(k=2n+u\). The first \(k\) primes can be enumerated in polynomial time and satisfy \(p_{\max}=O(k\log k)\) \citep{rosser1962primes}. An \(X\)-helper reads two coordinates, a \(Y\)-validity pattern one, and a clause pattern at most six. Their moduli are therefore bounded by \(p_{\max}^2,p_{\max},p_{\max}^6\), respectively. Since a full gadget contains one task per modulus position, the total task count is at most
\[ 1+4n p_{\max}^2+2u p_{\max}+7m p_{\max}^6+2n, \] 
which is polynomial in the source encoding length.

The definitions make \(q,S,W_0,B_{\mathrm{mand}},\Lambda,h\) polynomially bounded. Full-gadget parameters are \(C=w,T=b\Lambda\), and either \(D=r\Lambda\) or \(D=a\Lambda+h-1\), where \(a<b\), \(1\le r\le b\le p_{\max}^6\), and \(w\le\max\{W_0,S\}\). Selector parameters are \((qp,(p-1)\Lambda+h-1,p\Lambda)\), with \(p\le p_{\max}\). Therefore all timing parameters have polynomial magnitude. Values are \(V_0=2n+1\) or \(1\), and \(K=MV_0+n\) is polynomially bounded by the task-count bound.

The identity \(\Lambda-h=3n+u+6m+nq\ge0\) gives \(h\le\Lambda\). Since \(n\ge1\), \(S\ge2\), and \(W_0\ge1\), also \(h\ge2\) and \(h-1\ge\max\{W_0,S\}\). Every full-gadget deadline is at least \(h-1\), so \(C\le D\). Non-special deadlines satisfy \(D=r\Lambda\le b\Lambda=T\); special deadlines satisfy \(D=a\Lambda+h-1<(a+1)\Lambda\le b\Lambda=T\). For a selector, \(p\ge2\) gives 
\[ C=qp\le q p_{\max}<W_0\le\Lambda\le D<T=p\Lambda. \]
Therefore every generated task satisfies \(C\le D\le T\).

Computing the local CRT residues and expanding these polynomially bounded gadgets takes polynomial time. Since every output integer has polynomial magnitude, the same holds for unary output.
\end{proof}

The bound can still be large, though: the clause layer contains at most \(7m p_{\max}^6=O(m[(2n+u)\log(2n+u)]^6)\) tasks. The construction is only for the complexity proof and is not intended as a practical instance generator.

\section{Main Result}
\label{sec:main-result}

\begin{theorem}
\label{thm:completeness}
\(\VTS\) is \(\Sigma_2^P\)-complete for both constrained and arbitrary deadlines, even when all numerical parameters are encoded in unary.
\end{theorem}

\begin{proof}
Membership follows from Theorem~\ref{thm:membership}, including unary encoding. For hardness, reduce from the source problem of Section~\ref{sec:reduction-overview} using the construction of Section~\ref{sec:construction}. Lemmas~\ref{lem:forward} and~\ref{lem:reverse} establish logical equivalence. Lemma~\ref{lem:bounded-construction}
establishes a polynomial-time construction with valid constrained-deadline tasks, polynomial task count, and polynomial numerical parameters, so the reduction is polynomial under binary or unary encoding. This proves constrained-deadline completeness. The same instances are allowed in the arbitrary-deadline model, establishing hardness there as well.
\end{proof}

\section{Discussion}
\label{sec:discussion}

Unless the polynomial hierarchy collapses, \(\VTS\) belongs to neither NP nor coNP. Hardness persists with polynomially bounded numerical parameters, so an exact algorithm polynomial in task count and parameter magnitudes would imply \(\mathrm P=\Sigma_2^P\) and collapse the hierarchy to \(\mathrm P\). The disabled-option MIS model of Section~\ref{sec:related-work} is likewise \(\Sigma_2^P\)-complete. Approximation, online selection, multiprocessor platforms, and tractable restricted families remain separate questions.

\section{Conclusion}
\label{sec:conclusion}

We proved that \(\VTS\) for sporadic tasks on one preemptive processor is \(\Sigma_2^P\)-complete for both constrained and arbitrary deadlines, including unary encoding. CRT demand gadgets represent universal responses, while local selector locks enforce one existential truth choice per variable. The construction uses constrained deadlines and has polynomial task count and numerical parameters.

For comparison, the analogous selection problem under a supplied fixed task-priority order is NP-complete for implicit or constrained deadlines. Membership follows by guessing the selected set and response-time fixed points no larger than the corresponding deadlines; these certificates can be checked in polynomial time. For hardness, assign every task value \(1\) and set \(K=|\Gamma|\), forcing all tasks to be selected. NP-hardness then follows from fixed-priority schedulability
\citep{ekberg2017fixed}.

\section*{Funding}

This work was supported in part by the U.S. National Science Foundation under Grants 2211641, 2246672, 2442078, and 2502855, and by the Swedish Research Council under Grant 2023-04586.

\section*{Author contributions}

The original research problem was conceived and formulated by K.Y. and Z.G. C.H. identified the relevant complexity class, developed the \(\Sigma_2^P\)-completeness proof, and prepared the manuscript. All authors contributed to refining the proof strategy, strengthening the motivation and presentation of the technical results, and clarifying the scope of the mathematical claims. All authors critically reviewed the manuscript and its mathematical arguments, approved the final version, and agreed to be accountable for the integrity of the work.

\section*{Use of generative AI}

OpenAI Codex was used to assist with manuscript proofreading,
LaTeX formatting tweaks, and debugging of text-based explanatory diagrams in
TikZ. The authors reviewed all outputs and take responsibility for the accuracy of all output. No generative image model was used.
All figures are original diagrams compiled from LaTeX/TikZ source. No LLM was used in the generation of the content of the paper itself in any sort of generative context, only proofreading and LaTeX/TikZ debugging.

\appendix

\section{Source Normal Form}
\label{app:normal-form}

The reduction assumes a source instance of the form
\[
  \exists X\forall Y\,\varphi(X,Y),
\]
where \(\varphi\) is a 3-DNF formula whose terms have exactly three literals over distinct variables. This is a standard \(\Sigma_2^P\)-complete normal form \citep{stockmeyer1976ph,wrathall1976complete,schaefer2002compendium}. One route is to start from the complement of the usual \(\forall\exists\)-3CNF problem. Negating the inner CNF produces a DNF, and the leading quantifiers swap to the required \(\exists\forall\) form.
  
Terms with fewer than three literals can be padded using fresh universally quantified variables. If a term \(\theta\) has \(r<3\) literals, introduce \(3-r\) fresh variables \(z_1,\ldots,z_{3-r}\) and replace \(\theta\) by the disjunction over all assignments to those variables:
\[
  \bigvee_{\pi\in\{0,1\}^{3-r}}
   \left(
   \theta \wedge \bigwedge_{j=1}^{3-r} z_j^{\pi_j}
   \right),
\]
where \(z_j^1=z_j\) and \(z_j^0=\neg z_j\). For any assignment to the fresh variables, exactly one pattern \(\bigwedge_j z_j^{\pi_j}\) is true, so the padded disjunction is true exactly when \(\theta\) is true, and the formula is preserved. Contradictory terms are deleted, and duplicate literals inside a term are removed before padding. Degenerate always-true or always-false cases can be mapped explicitly to fixed instances: the one-task system \(\Gamma=\{(1,1,1,1)\}\) with \(K=1\) is a yes-instance, while the same system with \(K=2\) is a no-instance.

\bibliography{vostcds-references}

\end{document}